\documentclass{article}

\usepackage{lineno,hyperref}
\modulolinenumbers[5]

\usepackage[utf8]{inputenc}
\usepackage{amsmath}
\usepackage{amsthm}
\usepackage{amssymb}
\usepackage{graphicx}
\usepackage[colorinlistoftodos]{todonotes}
\usepackage{geometry}
\usepackage[export]{adjustbox}
\usepackage{algorithm}
\usepackage[noend]{algpseudocode}
\usepackage{subfigure}
\usepackage{natbib}

\usepackage{tikz}
\usetikzlibrary{shapes,arrows,fit,calc,positioning,patterns,decorations.pathmorphing,decorations.pathreplacing}

\tikzset{ brokenrect/.style={
    append after command={
      \pgfextra{
      \path[draw,#1]
       decorate[decoration={zigzag,segment length=0.3em, amplitude=.7mm}]
       {(\tikzlastnode.north east)--(\tikzlastnode.south east)}      
        -- (\tikzlastnode.south west)|-cycle;
        }}}}
\tikzset{ brokenrect2/.style={
    append after command={
      \pgfextra{
      \path[draw,#1]
       decorate[decoration={zigzag,segment length=0.3em, amplitude=.7mm}]
       {(\tikzlastnode.north west)--(\tikzlastnode.south west)}      
        -- (\tikzlastnode.south east)|-cycle;
        }}}}

\newtheorem{observation}{Observation}
\newtheorem{proposition}{Proposition}
\newtheorem{lemma}{Lemma}
\newtheorem{definition}{Definition}
\newtheorem{theorem}{Theorem}

\newcommand{\Q}{\mathcal{Q}}
\newcommand{\J}{\mathcal{J}}
\newcommand{\Fmax}{F_{\max}}

\begin{document}


\title{An online joint replenishment problem combined with single machine scheduling}



\author{P\'eter Gy\"orgyi\footnote{Institute for Computer Science and Control, E\"{o}tv\"{o}s Lor\'{a}nd Research Network, Kende Str. 13-17., 1111 Budapest, Hungary, e-mail: gyorgyi.peter@sztaki.hu}, ~Tam\'as Kis\footnote{Institute for Computer Science and Control, E\"{o}tv\"{o}s Lor\'{a}nd Research Network, e-mail: kis.tamas@sztaki.hu} ~and T\'imea Tam\'asi\footnote{Institute for Computer Science and Control, E\"{o}tv\"{o}s Lor\'{a}nd Research Network and Department of Operations Research, E\"otv\"os University, Budapest, Hungary, e-mail: tamasi.timea@sztaki.hu}}
%
%
%

\maketitle

\begin{abstract}
This paper considers a combination of the joint replenishment problem with single machine scheduling. 
There is a single resource, which is required by all the jobs, and a job can be started at time point $t$ on the machine if and only the machine does not process another job at $t$, and the resource is replenished between its release date and $t$. Each replenishment has a cost, which is independent of the amount replenished. The objective is to minimize the total replenishment cost plus the maximum flow time of the jobs.

We consider the online variant of the problem, where the jobs are released over time, and once a job is inserted into the schedule, its starting time cannot be changed. We propose a deterministic 2-competitive  online algorithm for the general input. Moreover, we show that for a certain class of inputs (so-called $p$-bounded input), the competitive ratio of the algorithm tends to $\sqrt{2}$ as the number of jobs tends to infinity. We also derive several lower bounds for the best competitive ratio of any deterministic online algorithm under various assumptions.

\end{abstract}
%
%


\section{Introduction}
In this paper, we study a combination of the classical joint replenishment problem (JRP) with machine scheduling, proposed recently by \citet{gyorgyi2021}.
The joint replenishment problem seeks an optimal replenishment policy of one or several items required to fulfill a sequence of demands over time.
When combined with machine scheduling, a demand is fulfilled only after the required item is replenished, and, in addition, processed on a machine for a given amount of time. The machine processes the demands in some order that has to be determined. The cost to be minimized has two main components: one is related to the replenishment of the items, and another to the scheduling of the demands on the machine. An example for the former one is a fixed cost due each time some item is replenished, while a possible scheduling related cost is the maximum flow time, which is the maximum difference between the completion time of a demand on the machine and its arrival time.
The processing of the demands on the machine adds an extra twist to the problem, and it may delay the fulfillment of the demands.  

In the scheduling literature, a machine processes {\em jobs\/}, and we will identify the demands with jobs. Likewise, we will say that a job $j$ has a {\em release date\/} $r_j$, which is the arrival time of the corresponding demand, and a {\em processing time\/} $p_j$, which equals the processing time of the demand on the machine. The time point when  the machine completes some job $j$ is called the  {\em completion time\/} of  job $j$, and is denoted by $C_j$. A {\em schedule\/} specifies a completion time $C_j$ for each job $j$, and we assume that $C_j \geq r_j + p_j$ for each job $j$, and for each pair of jobs $j \neq k$, either $C_j \geq C_k + p_j$ or $C_k \geq C_j + p_k$, that is, the two jobs are processed in non-overlapping time slots of the machine.

In \citet{gyorgyi2021}, a number of variants of the problem are studied which differ in the scheduling objective, and in the additional constraints on the processing times of the jobs, or in the frequency of the arrival of the jobs.
Offline and online algorithms are proposed with various performance guarantees. In the online problems studied the scheduling objective was the total completion time, the total flow time, and the maximum flow time.
However, in the variant with the maximum flow time objective, it was assumed that a job is released at every non-negative integer time point until no more jobs arrive, each job has a  processing time of one time unit, and the only unknown parameter is when the last job arrives. For this online problem, a $\sqrt{2}$-competitive algorithm has been proposed.

In the present work we focus uniquely on the online problem with the maximum flow time scheduling objective.
We propose a deterministic  2-competitive online algorithm for the input with demands arriving at arbitrary distinct integer time points. Furthermore, we analyze the performance of the algorithm on restricted input, where the difference of the arrival time of any two consecutive demands is bounded by some parameter $p$, so-called \emph{$p$-bounded input\/}. In this setting, the competitive ratio of our algorithm tends to $\sqrt{2}$ as the number of demands tends to infinity. We also provide some lower bounds for the best possible competitive ratio for any online algorithm for general as well as \emph{$p$-regular input\/}, where a demand arrives every $p$ time units.

{\it Notation.\/}
Throughout the paper, we will use the well-known $\alpha|\beta|\gamma$ notation introduced by \citet{graham1979} for classifying scheduling problems, where the $\alpha$ field describes the processing environment, the $\beta$ field consists of the additional restrictions and extensions, while the $\gamma$ field provides the objective function.
We consider single machine scheduling problems, which is denoted by 1 in the $\alpha$ field. In the $\beta$ field $r_j$ indicates  that the jobs have  release dates, while $p_j=1$ restricts the processing time each job to be one time unit.  In the $\gamma$ field, $\sum w_j C_j$ indicates the total weighted completion time objective, where the $w_j$ are non-negative job weights, $\Fmax$ the maximum flow time, which equals $\max F_j$ with $F_j = C_j - r_j$, $L_{\max}$ the maximum lateness defined as $\max L_j$ with $L_j = C_j - d_j$, where the $d_j$ are the due dates of the jobs.
All these objectives are to be minimized in the respective machine scheduling problems.
We will extended the $\beta$ field by $jrp$ which indicates that we combine machine scheduling with joint replenishment, and $s=1$ means that there is only one item type required by all the jobs (demands). Moreover, $\textit{distinct\ } r_j$ stipulates that the jobs have distinct release dates, i.e., $r_j \neq r_k$ for each pair of distinct jobs $j$ and $k$.
So, $1|jrp, s=1, \textit{distinct\ } r_j|\Fmax$ is a concise notation for the combined joint replenishment and scheduling problem on a single machine, where there is one item type, the jobs have distinct release dates, and the objective function is the maximum flow time.

{\it Related work\/}. The joint replenishment problem has been studied for more than 50 years, see \citet{khouja2008} for an overview. In the simplest version of JRP, a demand is ready as soon as the required items are replenished.
In other variants, such as JRP-D, 
the demands have deadlines, and the ordering cost is the only objective function.
In this paper we deal with a variant, where the fulfillment of the demands may be delayed, and the objective is to minimize the total cost incurred by late delivery and by the replenishments. This is called JRP-W, and it is strongly NP-hard even in case of linear delay cost functions, which follows from the NP-hardness of another variant examined by \citet{arkin1989} (called JRP-INV), by reversing the time line. Later, \citet{nonner2009} proved the NP-hardness of a more restricted variant, where each item admits only three distinct demands over the time horizon.

\citet{buchbinder13} described a 3-competitive algorithm for the online  problem with linear delay function, and they gave a lower bound of 2.64 for the best competitive ratio of any online algorithm. The latter was strengthened in \citet{bienkowski2014} to $2.754$, and the authors also proposed a 1.791-approximation algorithm for the offline problem.

There are a lot of results for offline and online   single machine scheduling problems for  different optimization criteria such as $\sum C_j$ \citep{lenstra1977,afrati1999,chekuri2001b}, $\sum w_j C_j$ \citep{anderson2004, hoogeveen1996,goemans2002} or $\sum F_j$ \citep{kellerer1999, epstein2001}. The $\Fmax$ objective is a special case of the $L_{\max}$ objective, where the goal is to minimize maximum lateness \citep{lageweg1976}. \citet{jackson1955} showed that the optimal solution for the problem $1||L_{\max}$ can be obtained by scheduling the jobs in non-decreasing order of their due dates, called the EDD rule. On the other hand, \citet{lenstra1977} showed that the problem $1|r_j|L_{\max}$ is strongly NP-hard.

For $1|r_j|L_{\max}$, \citet{hall1992} proposed an $O(n^2 \log n)$ algorithm with a worst-case ratio of $4/3$, and they also described two Polynomial Time Approximation Schemes. The online variant is analyzed by \citet{hoogeveen1996}. They showed that no online algorithm can obtain a competitive ratio better than $(\sqrt{5}+1)/2$. Using the EDD rule whenever the machine becomes idle leads to a 2-competitive algorithm, and by introducing a clever waiting strategy, the competitive ratio reaches the lower bound of $(\sqrt{5}+1)/2$.

\citet{gyorgyi2021} analyzed the combination of the joint replenishment problem (JRP-W) and the single machine scheduling with different scheduling objective such as $\sum C_j, \sum w_jC_j, \sum F_j$ and $\Fmax$. For the latter, the authors showed that if there are two resources, the problem is NP-hard even under very strong assumptions. For $1|jrp,s=1,r_j|\Fmax$ and $1|jrp,s=const,p_j=p,r_j|\Fmax$, polynomial algorithms based on  dynamic programming were proposed. The paper also considered some online variants of the problem with the $\sum w_j C_j$, $\sum F_j$, and $\Fmax$ objectives. For the former two objectives, deterministic 2-competitive online algorithms were proposed, while for the $\Fmax$ objective, only a special case was considered where a job arrives every time unit till an unknown time moment, so-called \emph{regular input\/}.
For the latter online problem, a deterministic  $\sqrt{2}$-competitive  algorithm was described, and it was shown that there is no deterministic $(4/3-\varepsilon)$-competitive algorithm for any $\varepsilon>0$. For the general input, it was shown that no deterministic  online algorithm can achieve a competitive ratio better than $(\sqrt{5}+1)/2$.

{\it Organization of the paper\/}. We provide the problem formulation and an overview of our results in Section~\ref{sec:prob_form}. In Section~\ref{sec:offline}, we present some properties of the offline optimum for later use. In Section~\ref{sec:online_general}, we propose a deterministic 2-competitive online algorithm for the general input, and also prove that for $p$-bounded input, the  competitive ratio of the same algorithm tends to $\sqrt{2}$. In Section~\ref{sec:num_res} we present numerical results regarding to the algorithm proposed in Section~\ref{sec:online_general}. In Section~\ref{sec:online_LB}, we derive lower bounds for the best competitive ratio for the general and $p$-regular input, respectively. We conclude the paper in Section~\ref{sec:concl}.

\section{Problem formulation and overview of main results}\label{sec:prob_form}
There is a set of $n$ jobs $\J$, one resource, and a single machine. 
Each job $j$ has a processing time $p_j = 1$, and a release date $r_j$. The release dates are distinct i.e. $r_j \neq r_k$ if $j \neq k$. A job can be processed on the machine from time $t$ only if there is a replenishment from the resource in $[r_j,t]$. Each replenishment incurs a cost of $K$.  All data is integral.

A \emph{solution} of the problem is a pair $(S,\Q)$, where $S = \{S_j, j \in \J \}$ is a schedule specifying the starting times of the jobs, and $\Q = \{\tau_i, 1\leq i\leq q\}$ is the set of replenishment times of the resource such that $\tau_i$ is the  $i^{th}$ replenishment time, and $\tau_i < \tau_{i+1}$ for each $i=1,\ldots,q-1$.
The solution is feasible, if the jobs do not overlap in time, i.e. $S_j+1 \leq S_k$ or $S_k+1 \leq S_j$ holds for each $j \neq k$, and for each job $j \in \J$ there exists some $\tau_i \in \Q$ such that $\tau_i \in [r_j,S_j]$. 
The completion of  job $j$ in schedule $S$ is $C_j=S_j+1$, and its flow time is $F_j = C_j - r_j$.
The {\em replenishment cost\/} of a solution is $c_{\Q} := Kq$, while the {\em maximum flow time\/} is $F_{\max}=\max_{j\in\J} F_j$.
The {\em cost of a solution\/} is $cost(S,\Q) = c_{\Q} + F_{\max}$. We seek  a feasible solution of minimum cost.

In the online problem, the jobs arrive over time, and there is no information about them before their release date. The solution is constructed step-by-step, the starting time of a job and the replenishment times, once fixed, cannot be reversed. 
However, upon arrival of the last job, the scheduler is notified immediately that there will be no more jobs.  

An input is called \emph{$p$-regular\/}, if $r_j=(j-1)p$ for $j\geq 1$, for a given integer $p\geq 1$. It is \emph{regular\/} if it is 1-regular.
 We will also consider \emph{$p$-bounded input}, where the only known information about the input is that the difference of two consecutive release dates is upper bounded by some number $p$, i.e., $r_{j+1}-r_j \leq p$ for $j \geq 1$.

We illustrate the problem and its possible  solutions in Example~\ref{ex}.

\paragraph{Example}\label{ex}
Consider an input consisting of two jobs, with release dates $r_1=0$ and $r_2=t$ for some $t \geq 1$. If an algorithm makes two replenishments in $\tau_1 < t$ and $\tau_2 = t$, then the cost of this solution is $cost(S,\Q) = 2K + \tau_1+1$. However, if we postpone the replenishment and the starting time of the first job, then the objective is $cost(S', \Q') = K+t+1$. We have saved $K$ at the replenishment cost, but the maximum flow time has increased by $t-\tau_1$. Depending on whether $K$ or $t-\tau_1$ is bigger, the first or the second solution has a smaller cost. See Figure~\ref{fig:example} for an illustration. Of course, in an online setting we do not know when the second job is released, therefore, we  have to make the decision about the first replenishment time $\tau_1$ before time $t$.

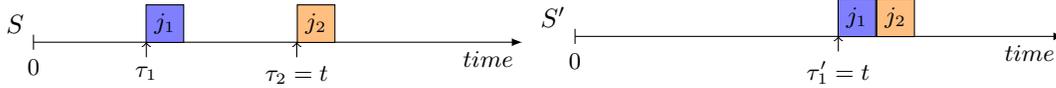
\begin{figure}[h]
\begin{tikzpicture}
\def\ox{0} 
\def\oy{0} 
\def\ui{1.5}
\def\uii{3.50}
\def\uiii{8}
\coordinate(o) at (\ox,\oy); 
\coordinate(u1) at (\ui,\oy);
\coordinate(u2) at (\uii,\oy);
\coordinate(u3) at (\uiii,\oy);

\tikzstyle{mystyle}=[draw, minimum height=0.5cm,rectangle, inner sep=0pt,font=\small]

\def\tl{6.5} 
\def\oyi{0}
\draw [-latex](\ox,\oyi) node[above left]{$S$} -- (\ox+\tl,\oyi) node[below left,font=\small]{$time$};

\draw[-] ($(o)-(0,-0.1)$) -- ($(o)-(0,0.1)$) node[below] {\small 0};
\draw[<-] (u1) -- ($(u1)-(0,0.2)$) node[below] {\small $\tau_1$};
\draw[<-] (u2) -- ($(u2)-(0,0.2)$) node[below] {\small $\tau_2 = t$};

\def\pi{0.7}
\node(b1) [above right=-0.01cm and -0.01cm of u1,mystyle, minimum width=0.5 cm, fill = blue!50]{$j_1$};
\node(b3) [above right=-0.01cm and -0.00cm of u2,mystyle, minimum width=0.5cm, fill = orange!50]{$j_2$};
\end{tikzpicture}
\begin{tikzpicture}
\def\ox{0} 
\def\oy{0} 
\def\ui{1.5}
\def\uii{3.50}
\def\uiii{8}
\coordinate(o) at (\ox,\oy); 
\coordinate(u1) at (\ui,\oy);
\coordinate(u2) at (\uii,\oy);
\coordinate(u3) at (\uiii,\oy);

\tikzstyle{mystyle}=[draw, minimum height=0.5cm,rectangle, inner sep=0pt,font=\small]

\def\tl{6.5} 
\def\oyi{0}
\draw [-latex](\ox,\oyi) node[above left]{$S'$} -- (\ox+\tl,\oyi) node[below left,font=\small]{$time$};

\draw[-] ($(o)-(0,-0.1)$) -- ($(o)-(0,0.1)$) node[below] {\small 0};
\draw[<-] (u2) -- ($(u2)-(0,0.2)$) node[below] {\small $\tau'_1=t$};

\def\pi{0.7}
\node(b1) [above right=-0.01cm and -0.01cm of u2,mystyle, minimum width= 0.5 cm, fill = blue!50]{$j_1$};
\node(b3) [right=0cm of b1,mystyle, minimum width=0.5cm, fill = orange!50]{$j_2$};
\end{tikzpicture}
\caption{Two feasible solutions. The arrows below the axis denote the replenishments.}\label{fig:example}
\end{figure}

In this article, we focus on the online problem $1|jrp,s=1,  p_j=1, \textit{distinct\ } r_j|\Fmax+c_\Q$. First, we present some preliminary results regarding the offline optimal solution. Although \citet{gyorgyi2021} already covered the offline variant of the problem, these results help in the analysis of the proposed online algorithm. We show that if the release dates of the jobs are not necessary distinct, then the problem $1|jrp, s=1, p_j=1, \textit{distinct\ }r_j|\Fmax+c_\Q$ is equivalent to the problem $1|jrp, s=1, r_j|\Fmax+c_\Q$ (i.e., when the jobs can have arbitrary processing times), see Proposition~\ref{prop:distinct}. This justifies our assumption that the jobs have distinct release dates.

We devise a deterministic 2-competitive online algorithm for the problem $1|jrp,s=1,p_j=1,distinct \ r_j|\Fmax+c_\Q$. For the so-called \emph{sparse input}, where the difference between two consecutive release dates $r_j$ and $r_{j+1}$ is lower bounded by $Kj$, this analysis is tight. On the other hand, we show that for the $p$-bounded input, the competitive ratio of the algorithm tends to $\sqrt{2}$ as the number of jobs tends to infinity. This result generalizes the one of \citet{gyorgyi2021} for the regular input, and although it does not reach the $\sqrt{2}$-competitive ratio for short sequences of jobs, in the long run it gets arbitrarily close to it.

Lastly, we provide new lower bounds for the best competitive ratio. For the general input, there is no online algorithm with competitive ratio of $3/2$, even if there are only two jobs in the input. In the case of three jobs, this lower bound is $4/3$. We also provide a lower bound for the best competitive ratio in the case of the $p$-regular input. For long sequences of jobs (i.e., where the last job arrives at some large time point $t > t_0$), there is no algorithm with a competitive ratio better than 1.015.

We mention that our online model  is slightly different from that of \citet{gyorgyi2021}. While in this paper, upon the arrival of the last job we get the information that there will be no further jobs, in \citet{gyorgyi2021} this information is not available at once. Therefore, the presented lower bounds cannot be directly compared with each other. In fact, we receive a smaller lower bound for the general case ($3/2$ instead of $(\sqrt{5}+1)/2$). 

We summarize our new results, along with some previous ones in Table~\ref{tab:results}.

\begin{table}[h]
\caption{Old and new results for the online problem $1|jrp,s=1,p_j=1,\textit{distinct\ } r_j| \Fmax$.}
\label{tab:results}
\begin{center}
\begin{tabular}{ccc}
\hline
Restriction & Result & Source\\

\hline
regular $r_j$	& $\sqrt{2}$-comp. alg.	&\citet{gyorgyi2021} \\
-	&no $(\sqrt{5}+1)/2$-comp.	alg.	&\citet{gyorgyi2021} \\
regular $r_j$		&no $4/3$-comp.	alg.	&\citet{gyorgyi2021} \\
\hline
-		&2-comp.	alg.	&Theorem~\ref{thm:2-comp} \\
$p$-bounded $r_j$	& $\sqrt{2}$-comp.~alg.~for $n\rightarrow \infty$	&Theorem~\ref{thm:sqrt2_comp} \\
 $n=2$		&no $3/2$-comp.	alg.	&Theorem~\ref{thm:neg_res_2jobs} \\
 $n \geq 3$	& no $4/3$-comp. alg.	&Theorem~\ref{thm:neg_res_3jobs} \\
$p$-regular $r_j$ 	&no $1.015$-comp.~alg.~for $n \rightarrow \infty$	 &Theorem~\ref{thm:neg_res} \\
\hline
\end{tabular}
\end{center}
\end{table}

\section{Properties of the offline optimum}\label{sec:offline}

In this section we are going to present some properties of the offline optimal solution regarding the general, $p$-regular and $p$-bounded input. Denote with $OPT(I)$ the offline optimum for input $I$. First we make some easy observations:

\begin{observation}\label{obs:repl_times_increasing}
It is enough to consider solutions for which the jobs are scheduled in increasing order of their release date.
\end{observation}

\begin{observation}\label{obs:repl_times_rel_dates}
It suffices to replenish the resource only at the release dates of some jobs \citep{gyorgyi2021}.
\end{observation}

\begin{observation}\label{obs:subset}
If $I' \subseteq I$, then $OPT(I') \leq OPT(I)$.
\end{observation}

In this article, we only consider inputs where the release dates are distinct. If this condition does not hold, the problem is equivalent to a more general problem:

\begin{proposition}
If the release dates are not distinct, then the problem $1|jrp, s=1,p_j=1,r_j| \Fmax$ is equivalent to the problem $1|jrp,s=1,r_j| \Fmax$, i.e., when the jobs have arbitrary processing time.\label{prop:distinct}
\end{proposition}

\begin{proof}
Consider an optimal solution $(S^\star, \Q^\star)$ for the input $I$ of the problem $1|jrp,s=1,p_j=1,r_j| \Fmax$. Let $I_t$ be the set of jobs released at time $t$ in $I$. By Observation~\ref{obs:repl_times_increasing}, the jobs are ordered in non-decreasing order of their release dates in $(S^\star, \Q^\star)$, hence, jobs in $I_t$ are scheduled consecutively after each other. Let $I'$ be the following input for the problem $1|jrp,s=1, r_j| \Fmax$: for each $t$, for which $I_t \neq \emptyset$, define a job $j$ such that $r_j = t$ and $p_j = |I_t|$. Observe that the solution $(S^\star, \Q^\star)$ is feasible for $I'$.

For the other direction, consider an input $I'$ for the problem $1|jrp,s=1,r_j| \Fmax$, and let $(S', \Q')$ be an optimal solution. Now let the input $I$ for the problem $1|jrp,s=1,p_j=1,r_j| \Fmax$ be the following: for each $j \in I'$, define $p_j$ unit-length jobs with release date $r_j$. Similarly, $(S', \Q')$ is a feasible solution for $I$. Therefore, the two problems are equivalent, hence the statement is proved.
\end{proof}

From now on, we only consider the problem $1|jrp,s=1,p_j=1,distinct \ r_j| \Fmax$.

\begin{observation} \label{obs:block_flowtime}
For any feasible schedule, consider any pair of two jobs, $j$ and $k$ (for which $r_j<r_k$), scheduled consecutively, i.e., $S_j+1 = S_k$, on the machine.
Then $F_j \geq F_k$.
\end{observation}

From this observation, we can conclude the following:

\begin{observation} \label{obs:fmax_idle}
For any feasible solution, the maximum flow time is given by the first job scheduled in some $\tau_i$, such that the machine is idle before $\tau_i$.
\end{observation}

Following Observations~\ref{obs:repl_times_increasing} and \ref{obs:repl_times_rel_dates}, we only consider offline solutions where the replenishments occur at the job release dates, and the jobs are scheduled in increasing release date order as soon as possible (i.e., after each of the earlier jobs are scheduled and after the first replenishment following their release date).

Next, we derive some properties of the $p$-regular input consisting of $n$ jobs, which we denote by $R_n$.

\begin{observation}\label{obs:flowtime_from_replenishment}  Consider the $p$-regular input $R_n$.
\begin{enumerate}
\item [i)]If there is a replenishment which provides resource for at least $n'$ jobs in a feasible solution, then the maximum flow time is at least $(n'-1)p+1$.
\item [ii)]For any feasible solution with $q$ replenishments, there exists a job which has a flow time of at least $(\lceil n/q\rceil-1)p+1$.
\item [iii)]The cost of any feasible solution with $q$ replenishments is at least  $qK+(\lceil n/q\rceil-1)p+1$.

\end{enumerate}
\end{observation}

\begin{proof}
Let $\tau$ be the time of the replenishment, and denote with $f$ and $\ell$ the first and the last job in the schedule for which the replenishment is in $\tau$. Then, $\tau \geq r_{\ell} = r_f+(n'-1)p$, from which $i)$ follows.

If there are $n$ jobs, then by pigeonhole principle, there is a replenishment which provides resources for at least $\lceil n/q\rceil$ jobs. Then, $ii)$ follows from $i)$.

$iii)$ follows directly from $ii)$.
\end{proof}

\begin{proposition}\label{prop:opt_q_repl}
For the $p$-regular input, the minimum cost of any solution with  $q$ replenishments is  $qK+(\lceil n/q\rceil-1)p+1$.
\end{proposition}
\begin{proof}
We construct a feasible solution with $q$ replenishments and  total cost as claimed.

Let $r$ be such that $n=q\lfloor n/q \rfloor + r$.
Let 

\[ \tau_i :=
  \begin{cases}
  -p, & \text{ if } i=0,\\
  i p\left\lceil n/q \right\rceil, &  \text{ if } i\in \{1,\ldots,r\}, \\[1ex]
    i p\left\lfloor n/q \right\rfloor, &  \text{ if } i\in \{r+1,\ldots,q\},
  \end{cases}
\]
and schedule the jobs in increasing release date order as soon as possible.
Let $j$ be the $k^{th}$ job  ($k \geq 1$)  arriving after $\tau_i$, but not later than $\tau_{i+1}$ for $i=0,\ldots,q-1$. Then $\tau_i<r_j=\tau_i+kp\leq \tau_{i+1}$, and $C_j=\tau_{i+1}+k$, thus $F_j=(\tau_{i+1}-\tau_i)+k(1-p)\leq \lceil n/q\rceil p+k(1-p)$.
By Observation~\ref{obs:block_flowtime}, this expression is maximal if $k=1$. 
Therefore, $F_{\max}\leq (\lceil n/q\rceil-1)p+1$ and the cost of this solution is at most $qK+(\lceil n/q\rceil-1)p+1$. 
Equality follows from Observation~\ref{obs:flowtime_from_replenishment}.
\end{proof}
Now it follows immediately that
\begin{lemma}\label{lem:opt_balanced}
For  $p$-regular input, the offline optimum is
\[
OPT(R_n) = \min_{q\in \mathbb{Z}_{\geq 1}} \left(Kq + (\lceil n/q \rceil-1)p + 1\right).
\]
The $q^*$ giving the minimum value is the number of replenishments in an optimal solution.
\end{lemma}

Next we derive lower and upper bounds on the optimum for $p$-regular input. To this end, we define the function $f(q)$:
\[
f(q) = Kq + (n/q -1)p+ 1.
\] 
Note that $f(q)$ is quite similar to the expression for $OPT(R_n)$ in Lemma~\ref{lem:opt_balanced}.
\begin{lemma}\label{lem:opt_lb} For  $p$-regular input, $OPT(R_n) \geq \min_{q \in \mathbb{R}_{>0}} f(q) = 2\sqrt{npK}-p+1$.
\end{lemma}
\begin{proof}

By Lemma~\ref{lem:opt_balanced}, we can derive
\[
OPT(R_n) = Kq^{\star} + \left( \left \lceil {n}/{q^{\star}} \right \rceil -1 \right ) p + 1 \geq f(q^*) \geq \min_{q \in \mathbb{R}_{>0}} f(q).
\]
This expression is minimal if $q=\sqrt{np/K}$, for which we obtain the minimum value of $2\sqrt{npK}-p+1$. 
\end{proof}

\begin{lemma}\label{lem:opt_lb_diff}
For  $p$-regular input, $OPT(R_n) \leq 2\sqrt{npK} + K + 1$.
\end{lemma}

\begin{proof}
By Lemma~\ref{lem:opt_balanced}, there exists some $q^* \in \mathbb{Z}_{\geq 1}$ such that
\[
OPT(R_n) = Kq^{\star} + \left( \left \lceil {n}/{q^{\star}} \right \rceil -1 \right ) p + 1.
\]

It is easy to see that $f(q^{\star}) \leq OPT(R_n) \leq f(q^{\star})+p$. 

Let $\hat{q} = \sqrt{np/K} $ be the point minimizing $f(q)$ on the positive orthant.
Observe that $f(q)$ is a convex function, hence $|\hat{q}-q^{\star}| < 1$ holds. We distinguish two cases.

First assume $\hat{q} \leq q^{\star} < \hat{q}+1$. Then
\[f(q^{\star}) = Kq^{\star} + np/q^{\star} + 1 \leq K(\hat{q}+1) + np/\hat{q} + 1 = f(\hat{q})+K = 2 \sqrt{npK}-p+1 + K.\]
Hence, $OPT(R_n)\leq 2 \sqrt{npK}+K+1$.

Second, assume $q^{\star} < \hat{q} \leq q^{\star}+1$, and we verify that
\[
f(\hat{q}) = 2 \sqrt{npK}-p+1 \geq OPT(R_n)-p-K,
\]
from which the statement follows. To see this, we compute
\[
\begin{split}
f(\hat{q}) & = K\hat{q}+ (n/{\hat{q}} -1)p + 1 \geq K q^{\star} + (n/{(q^{\star}+1)} -1)p+ 1 \\
& \geq K (q^\star +1) + \left ( \left \lceil n/(q^\star+1)\right \rceil-1 \right )p + 1 -p - K \geq OPT(R_n)-p-K,
\end{split}
\]
where the first inequality follows from $q^\star < \hat{q} \leq q^\star+1$ by assumption, the second from the properties of integer rounding, and the  last from the definition of $q^\star$.
\end{proof}

Let $I$ be an $p$-bounded input, where the first job arrives in $t_{\min}$, and the last job arrives in $t$ (this means that the last job can be completed earliest in $t+1$). Consider the regular input and the $p$-regular input between $t_{\min}$ and $t$ denoted by $D_I$ and $R_I$, respectively. See Figure~\ref{fig:inputs} for an illustration of these three different inputs in the case of $p=3$.

\begin{figure}[h]
\centering
\begin{tikzpicture}

\def\ox{0} 
\def\oy{0} 

\def\tmin{1.2}
\def\tmax{6.7}

\coordinate(o) at (\ox,\oy); 
\coordinate(utmin) at (\tmin,\oy);
\coordinate(utmax) at (\tmax,\oy);

\tikzstyle{mystyle}=[draw, minimum height=0.5cm,rectangle, inner sep=0pt,font=\scriptsize]

\def\tl{10} 

\draw[] (0,0) node[above right] {$I$};

\draw[-] ($(utmin)-(0.2,-0.12)$) -- ($(utmin)-(0.2,0.08)$) node[below] {$t_{\min}$};
\draw[-] ($(utmax)-(0.2,-0.12)$) -- ($(utmax)-(0.2,0.08)$) node[below] {$t+1$};

\draw[-] ($(utmin)$) -- ($(utmax)-(0.2,0)$) node[ ] {$ $};

\draw [draw=black, fill = black!20] (1,0) rectangle (1.5,0.5) ;
\draw [draw=black, fill = black!20] (1.5,0) rectangle (2,0.5);
\draw [draw=black, fill = black!20] (2.5,0) rectangle (3,0.5);
\draw [draw=black, fill = black!20] (4,0) rectangle (4.5,0.5);
\draw [draw=black, fill = black!20] (4.5,0) rectangle (5,0.5);
\draw [draw=black, fill = black!20] (5,0) rectangle (5.5,0.5);
\draw [draw=black, fill = black!20] (6,0) rectangle (6.5,0.5);

\end{tikzpicture}

\begin{tikzpicture}

\def\ox{0} 
\def\oy{0} 

\def\tmin{1.2}
\def\tmax{6.7}

\coordinate(o) at (\ox,\oy); 
\coordinate(utmin) at (\tmin,\oy);
\coordinate(utmax) at (\tmax,\oy);

\tikzstyle{mystyle}=[draw, minimum height=0.5cm,rectangle, inner sep=0pt,font=\scriptsize]

\def\tl{10} 

\draw[] (0,0) node[above right] {$D_I$};

\draw[-] ($(utmin)-(0.2,-0.12)$) -- ($(utmin)-(0.2,0.08)$) node[below] {$t_{\min}$};
\draw[-] ($(utmax)-(0.2,-0.12)$) -- ($(utmax)-(0.2,0.08)$) node[below] {$t+1$};

\draw[-] ($(utmin)$) -- ($(utmax)-(0.2,0)$) node[ ] {$ $};

\draw [draw=black, fill = black!20] (1,0) rectangle (1.5,0.5);
\draw [draw=black, fill = black!20] (1.5,0) rectangle (2,0.5);
\draw [draw=black, fill = black!20] (2,0) rectangle (2.5,0.5);
\draw [draw=black, fill = black!20] (2.5,0) rectangle (3,0.5);
\draw [draw=black, fill = black!20] (3,0) rectangle (3.5,0.5);
\draw [draw=black, fill = black!20] (3.5,0) rectangle (4,0.5);
\draw [draw=black, fill = black!20] (4,0) rectangle (4.5,0.5);
\draw [draw=black, fill = black!20] (4.5,0) rectangle (5,0.5);
\draw [draw=black, fill = black!20] (5,0) rectangle (5.5,0.5);
\draw [draw=black, fill = black!20] (5.5,0) rectangle (6,0.5);
\draw [draw=black, fill = black!20] (6,0) rectangle (6.5,0.5);

\end{tikzpicture}

\begin{tikzpicture}

\def\ox{0} 
\def\oy{0} 

\def\tmin{1.2}
\def\tmax{6.7}

\coordinate(o) at (\ox,\oy); 
\coordinate(utmin) at (\tmin,\oy);
\coordinate(utmax) at (\tmax,\oy);

\tikzstyle{mystyle}=[draw, minimum height=0.5cm,rectangle, inner sep=0pt,font=\scriptsize]

\def\tl{10} 

\draw[] (0,0) node[above right] {$R_I$};

\draw[-] ($(utmin)-(0.2,-0.12)$) -- ($(utmin)-(0.2,0.08)$) node[below] {$t_{\min}$};
\draw[-] ($(utmax)-(0.2,-0.12)$) -- ($(utmax)-(0.2,0.08)$) node[below] {$t+1$};

\draw[-] ($(utmin)$) -- ($(utmax)-(0.2,0)$) node[ ] {$ $};

\draw [draw=black, fill = black!20] (1,0) rectangle (1.5,0.5);
\draw [draw=black, fill = black!20] (2.5,0) rectangle (3,0.5);
\draw [draw=black, fill = black!20] (4,0) rectangle (4.5,0.5);
\draw [draw=black, fill = black!20] (5.5,0) rectangle (6,0.5);

\end{tikzpicture}
\caption{The inputs $I$, $D_I$ and $R_I$ for $p=3$.}\label{fig:inputs}
\end{figure}
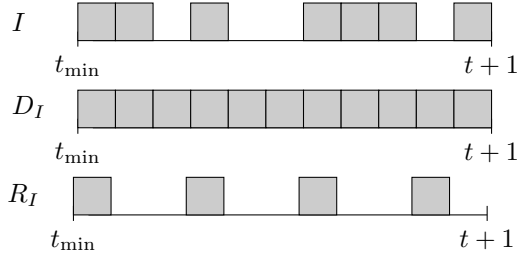

\begin{proposition}\label{prop:opt_R_I}
$OPT(D_I) \geq OPT(I) \geq OPT(R_I)$.
\end{proposition}

\begin{proof}
Consider an optimal solution  for the input $D_I$. Since $I \subseteq D_I$, by removing the jobs in $D_I \setminus I$ from this optimal solution, we obtain a feasible solution for $I$. Therefore $OPT(D_I) \geq OPT(I)$.

Now consider an optimal solution $(S^\star,Q^\star)$ for the input $I$, with maximum flow time of $F_{\max}^\star$. Observe that the number of jobs in $I$ is at least the number of jobs in $R_I$ (otherwise $I$ would not be a $p$-bounded input). Let $j_i$  and  $j_i'$ be the $i^{th}$ job in $R_I$ and $I$, respectively. It is easy to see that $r_{j_i} \geq r_{j_i'}$ for $i \in \{1, \ldots, |R_I|\}$.

We are going to create a feasible solution $(S,Q)$ for $R_I$ from the optimal solution for $I$: let $Q := Q^\star$, and $S_{j_i} := S^\star_{j_i'}$ for $i \in \{1, \ldots, |R_I|\}$. Since $r_{j_i} \geq r_{j_i'} \ \forall \ i$, this is indeed a feasible solution for $R_I$, with maximum flow time of at most $F_{\max}^\star$, and with the same replenishment cost. Therefore, $OPT(I) \geq OPT(R_I)$.
\end{proof}

\section{Online algorithm for the general input}\label{sec:online_general}

Consider Algorithm~\ref{alg:online_2comp_new}:
\begin{algorithm}
\caption{Online algorithm for the general input}
Initialization: $ t:= 0, F_{\max} := 0$.
\begin{enumerate}
\item Determine the set $B_t$ of unscheduled jobs at time $t$. \label{step:first_new}

\item Let $\Fmax^u$ be the maximum flow time of the jobs in $B_t$ if they are scheduled from $t$ in non-decreasing order of the release dates without gap.
\item If $\Fmax^u = \Fmax+K$, then replenish the resource, start the jobs of $B_t$ from $t$, $t := t+|B_t|$, and $\Fmax := \Fmax+K$. If the last job has already been scheduled, then STOP. 
\label{step:replenish_new}

\item If no job is scheduled at $t$, then $t := t+1$.
\item Go to step~\ref{step:first_new}.
\end{enumerate}
\label{alg:online_2comp_new}
\end{algorithm}

For an input $I$, denote the cost provided by Algorithm~\ref{alg:online_2comp_new} with $ALG(I)$, and let $q$ be the total number of replenishments, and
$\tau_1,\ldots,\tau_q$ the replenishment times.

Observe that $\tau_{i+1}-\tau_i >\tau_i-\tau_{i-1}$ holds for every $1\leq i<q$ (where $\tau_0 = 0$), i.e. jobs starting in $\tau_i$ are always finished before $\tau_{i+1}$. 

\begin{proposition}\label{prop:alg_ith_repl}
The cost of the algorithm at the $i^{th}$ replenishment is $2Ki$, with maximum flow time of $Ki$ for every $1 \leq i \leq q$.
\end{proposition}

\begin{proof}
We prove this by induction. At the first replenishment, the maximum flow time is $K$, hence the cost of the algorithm is $2K$. Suppose that at $\tau_{i-1}$, the algorithm has a cost of $2K(i-1)$ with maximum flow time of $\Fmax = K(i-1)$. The $i^{th}$ replenishment occurs when the maximum flow time of the jobs reaches $\Fmax + K = Ki$, while the total cost increases to $2Ki$.
\end{proof}

\begin{proposition} \label{prop:tau_i_diff}
$\tau_{i}-\tau_{i-1} \geq Ki$ for every $1 \leq i \leq q$.
\end{proposition}

\begin{proof}
If $j$ is the first job released after the $(i-1)^{th}$ replenishment, then $r_j \geq \tau_{i-1}+1$. The maximum flow time at the $i^{th}$ replenishment is given by the flow time of $j$, which is $Ki$. Therefore, $Ki = F_j = \tau_i+1-r_j  \leq \tau_i-\tau_{i-1}$.
\end{proof}

For the sake of analyzing the performance of Algorithm~\ref{alg:online_2comp_new}, we define a special class of inputs.

\begin{definition}
We call an input $I$ \emph{sparse}, if $r_{j+1}-r_j \geq Kj$ for all $1 \leq j < n$, where $n$ is the number of jobs in $I$.
\end{definition}

\begin{proposition} \label{prop:opt_sparse} If $I$ is a sparse input consisting of $n$ jobs, then $OPT(I) = Kn+1$.
\end{proposition}

\begin{proof}

Let $(S,\Q)$ be the feasible solution, where every job is replenished and scheduled at its release date (that is, $S_j=r_j$ for every $1 \leq j \leq n$, and $Q = \{r_1, \ldots, r_n\}$). We are going to show that $(S,\Q)$ is optimal, from which the  statement follows, since $cost(S,\Q) = Kn+1$.

By contradiction, assume that the solution $(S,\Q)$ is not optimal. Consider an optimal solution $(S^\star, \Q^\star)$ consisting of $n-x$ replenishments, where $1 \leq x < n$. This means that $x$ replenishments are removed from $\Q$, and the jobs scheduled at these replenishment times in $S$ are scheduled later in $S^\star$.

If for some $x < j \leq n$, there is a job $j$ such that $r_j \not \in \Q^\star$, then in $S^\star$, $j$ starts not sooner than $r_{j+1}$. Hence, the flow time of $j$ is not smaller than $r_{j+1}+1-r_j \geq Kj+1 > Kx+1$. Therefore, $cost(S^\star, \Q^\star) > K(n-x)+Kx+1=Kn+1 = cost(S,\Q)$, contradiction.

It follows that there is no such job $j$. Since there are $n-x$ replenishments and $r_{x+1}, \ldots, r_n \in \Q^\star$, then the first $x$ replenishment times of $(S,\Q)$ has to be all removed from $(S^\star, \Q^\star)$, and jobs $1, \ldots, x$ start from $r_{x+1}$. Then the flow time of the first job in $S^\star$ is $r_{x+1}+1-r_1 \geq Kx(x+1)/2 + 1$, and every other job has smaller flow time. Therefore, $cost(S^\star, \Q^\star) \geq Kx(x+1)/2 + K(n-x)+1 > Kn+1 = cost(S,\Q)$, contradiction. 

Hence, the optimal solution is $(S,\Q)$.
\end{proof}

We can also make an observation regarding the behaviour of Algorithm~\ref{alg:online_2comp_new} for the sparse input.

\begin{proposition}\label{prop:alg_sparse}
If $I$ is a sparse input consisting of $n$ jobs, then Algorithm~\ref{alg:online_2comp_new} replenishes $n$ times. The cost of the solution is $2Kn$.
\end{proposition}

\begin{proof}
We are going to show that each job is replenished individually from which the first statement follows. We proceed by induction on the job index. The first job is released at $r_1$, and Algorithm~\ref{alg:online_2comp_new} replenishes and starts this job at time $\tau_1 = r_1+K-1$. Since $\tau_1 < r_2$, the second job gets a separate replenishment.

Suppose that the $(j-1)^{th}$ job is replenished at $\tau_{j-1} < r_j$, by that time there are $j-1$ replenishments, and the maximum flow time is $K(j-1)$ by Proposition~\ref{prop:alg_ith_repl}. The next job is released at $r_j$, therefore, the algorithm is going to replenish and start this job when its flow time reaches $Kj$, i.e., $\tau_j= Kj + r_j - 1 \leq r_{j+1}-1$, where the last inequality follows from the definition of the sparse input.

By Proposition~\ref{prop:alg_ith_repl}, the cost of the solution is $2Kn$.
\end{proof}

\begin{theorem}
Algorithm~\ref{alg:online_2comp_new} is 2-competitive for the general input.\label{thm:2-comp}
\end{theorem}

\begin{proof}
Consider an input $I$ for which Algorithm~\ref{alg:online_2comp_new} makes $q$ replenishments in time points $\tau_1, \ldots, \tau_q$.
Then $ALG(I) = 2Kq$ by Proposition~\ref{prop:alg_ith_repl}. For $1 \leq i \leq q$, denote with $f_{i}$ the job from $I$ that starts at $\tau_{i}$. We define a new input with these  $q$ jobs, $I' = \{f_1, \ldots, f_q\}$. By Proposition~\ref{prop:tau_i_diff}, we have $\tau_{j+1}-\tau_j \geq Kj$, hence $I'$ is a sparse input. By Propositions~\ref{prop:opt_sparse} and \ref{prop:alg_sparse}, $OPT(I') = Kq + 1$ and  $ALG(I') = 2Kq$, hence $ALG/OPT = 2Kq/(Kq+1) \leq 2$.

On the other hand, by Observation~\ref{obs:subset}, $OPT(I) \geq OPT(I')$, and we have already noticed that $ALG(I) = 2Kq$ and $ALG(I') = 2Kq$. Consequently,  $ALG(I)/OPT(I) \leq ALG(I')/OPT(I')$. Hence, the competitive ratio is the largest for the sparse inputs, for which the algorithm is 2-competitive. This concludes the theorem. Observe that the analysis is tight for the sparse input.
\end{proof}

\begin{theorem}\label{thm:sqrt2_comp}
For the $p$-bounded input, the competitive ratio of Algorithm~\ref{alg:online_2comp_new} tends to $\sqrt{2}$ as the number of jobs tends to infinity.
\end{theorem}

\begin{proof}
Consider a $p$-bounded input $I$ consisting of $n$ jobs, for which Algorithm~\ref{alg:online_2comp_new} makes $q$ replenishments in $\tau_1, \ldots, \tau_q$.

Let $n_i$ be the number of jobs released between $\tau_{i-1}+1$ and $\tau_i$ for $1 \leq i \leq q$. Since the input is $p$-bounded, we obtain that:

\[n_i \geq \left \lceil (\tau_i-\tau_{i-1})/p \right \rceil \geq \left \lceil Ki/p \right \rceil \geq Ki/p,\]

where the second inequality follows by Proposition~\ref{prop:tau_i_diff}. Hence, 

\[n = \sum_{i=1}^{q} n_i \geq \sum_{i=1}^{q} Ki/p = Kq(q+1)/2p \geq Kq^2/2p,\]

from which $q \leq \sqrt{2pn/K}$ follows. Therefore:

\[ALG(I) \leq 2Kq \leq 2K \sqrt{2np/K} = 2\sqrt{2npK}.\]

On the other hand, by Lemma~\ref{lem:opt_lb} and Proposition~\ref{prop:opt_R_I}, we have 

\[OPT(I) \geq 2\sqrt{npK}-p+1.\]

It follows that

\[\frac{ALG(I)}{OPT(I)} \leq \frac{2\sqrt{2npK}}{2\sqrt{npK}-p+1} = \frac{\sqrt{2npK}}{\sqrt{npK}-p/2+1/2} \rightarrow \sqrt{2}, \text{ if } n \rightarrow \infty,\]
hence, the statement is proved.

The analysis is tight: consider the $p$-regular input $R_n$, which is also $p$-bounded. By Lemma~\ref{lem:opt_lb_diff},  $OPT(R_n) \leq  2 \sqrt{npK}+K+1$. Therefore,

\[\frac{ALG(R_n)}{OPT(R_n)} \geq \frac{2\sqrt{2npK}}{2 \sqrt{npK}+K+1} = \frac{\sqrt{2npK}}{\sqrt{npK}+K/2+1/2} \rightarrow \sqrt{2} \text{ if } n \rightarrow \infty.\]
\end{proof}

\section{Numerical results}
\label{sec:num_res}

In this section we analyse the competitive ratio proposed in Section~\ref{sec:online_general}. We proved that the algorithm obtains a competitive ratio of 2, which tends to $\sqrt{2}$ in the case of $p$-bounded inputs for some constant $p$. A question arises as to where does the competitive ratio lie if the difference of two consecutive release date follows some probability distribution $D$. 

Formally, we generate an input consisting of $n$ jobs, where $r_j = X_1+\ldots+X_j$ for $1 \leq j \leq n$, and $X_j$ is a random variable chosen from a discrete distribution $D$, with possible values of $1, 2, 3$ etc. Note that if the distribution $D$ has a finite support, then it is straightforward that the competitive ratio of such inputs tend to $\sqrt{2}$ as $n$ tends to infinity, since if $D$ is upper bounded by some finite number $p$, then the inputs generated this way are $p$-bounded with probability 1. Hence, we assume that $D$ has infinite support.

We chose $D$ to be a geometric distribution with parameter $\beta$, supported on the set $\{1,2,3,\ldots\}$. That is, $P(X_j = k) = (1-\beta)^{k-1}\beta$ for every $1 \leq j \leq n$. We fixed the replenishment cost to $K=1$. We generated 1000 instances for different values of $\beta$ and $n$. For an input consisting of $n$ jobs, we ran the offline algorithm and Algorithm~\ref{alg:online_2comp_new}, respectively, to obtain the competitive ratio. Figures~\ref{fig:comp_ratios_0.01}-\ref{fig:comp_ratios_0.0001} show the results of the experiments for $\beta \in \{0.01,0.001,0.0001\}$ and different values of $n$.

The smaller $\beta$ is, the closer the competitive ratio is to 2, since the input becomes sparse with high probability. On the other hand, by increasing the number of jobs, the competitive ratio is quickly decreasing, and it tends to the ratio $\sqrt{2}$ as in the case of a $p$-bounded input. This is due to the fact that the time between two consecutive release dates has an expected value of $1/\beta$, hence, if the number of jobs is significantly big, the input becomes $1/\beta$-bounded with high probability. If $\beta$ is relatively large, the competitive ratio is close to $\sqrt{2}$ even for small values of $n$, see Figure~\ref{fig:comp_ratios_0.01}. As the value of $\beta$ decreases, larger numbers of jobs are needed to approach the desired ratio of $\sqrt{2}$, see Figures~\ref{fig:comp_ratios_0.001} and \ref{fig:comp_ratios_0.0001}.
 
\begin{figure}[h!]
\centering
\includegraphics[width=.9\textwidth]{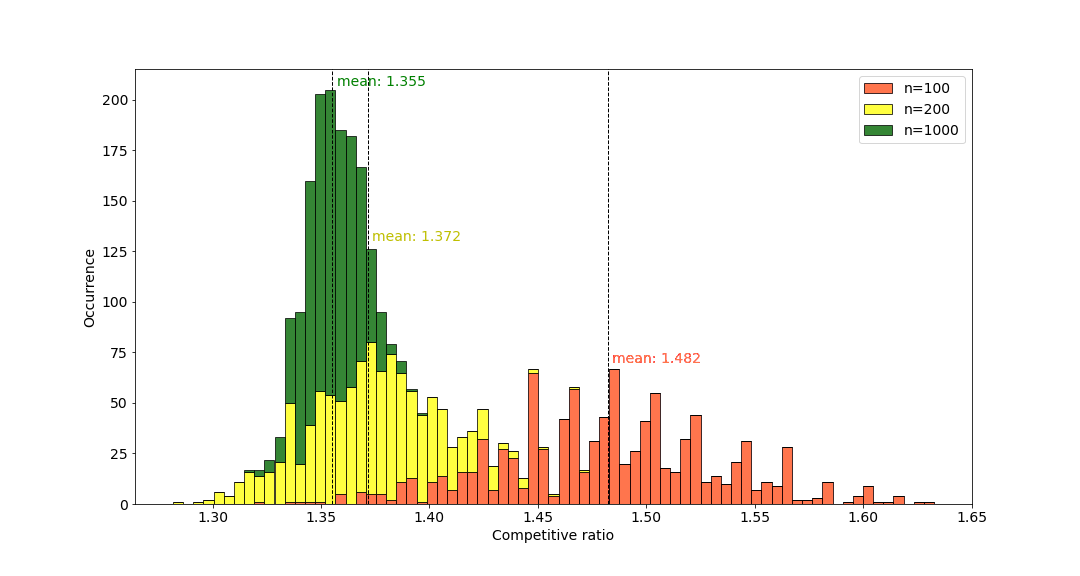}
\caption{Competitive ratios for $\beta = 0.01, n \in \{100,200,1000\}$.}\label{fig:comp_ratios_0.01}
\end{figure}

\begin{figure}[h!]
\centering
\includegraphics[width=.9\textwidth]{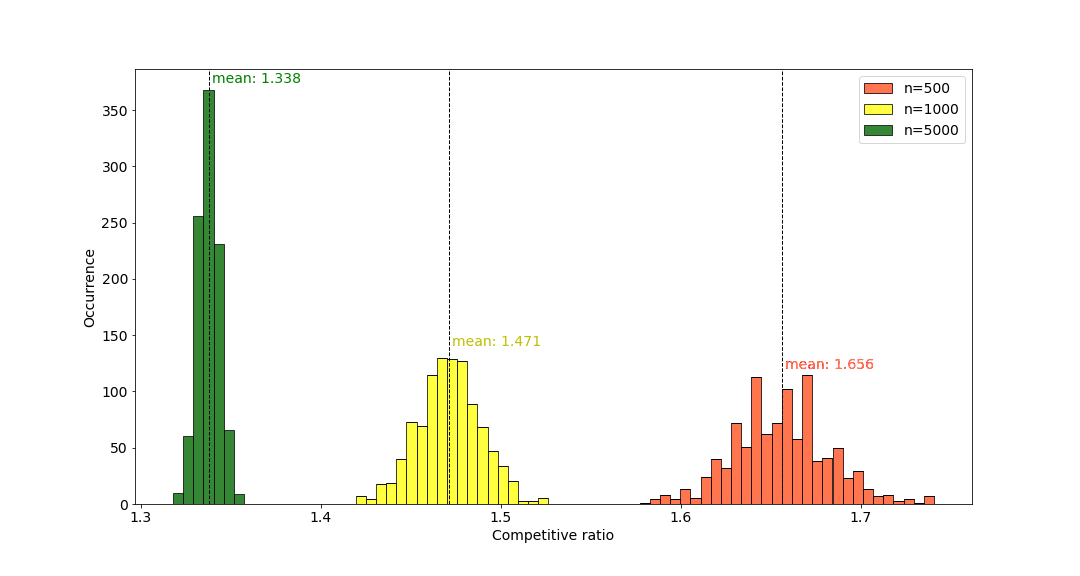}
\caption{Competitive ratios for $\beta = 0.001, n \in \{500,1000,5000\}$.}\label{fig:comp_ratios_0.001}
\end{figure}

\begin{figure}[h!]
\centering
\includegraphics[width=.9\textwidth]{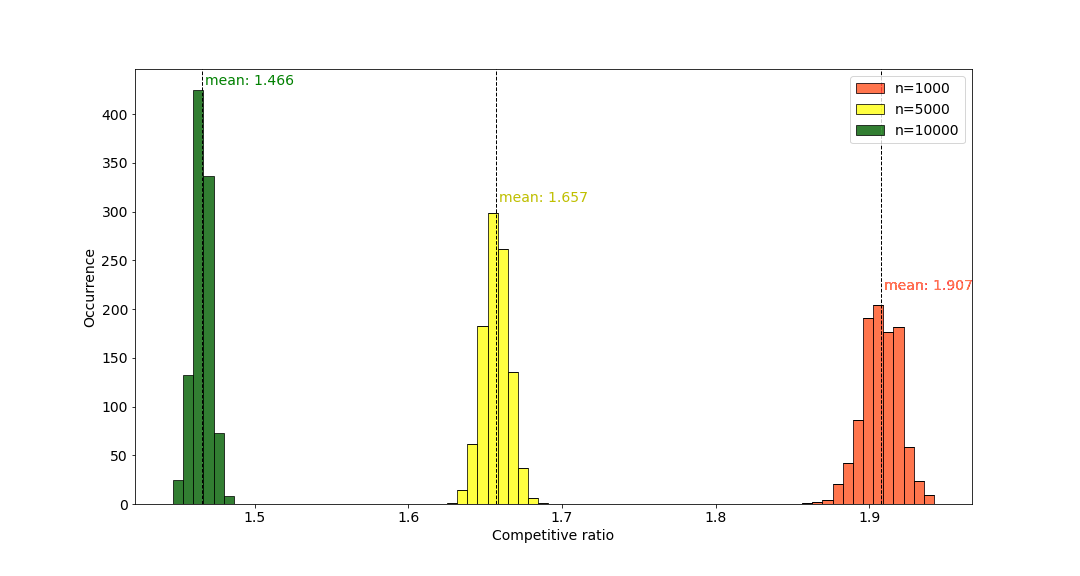}
\caption{Competitive ratios for $\beta = 0.0001, n \in \{1000,5000,10000\}$.}\label{fig:comp_ratios_0.0001}
\end{figure}


\section{Lower bounds for the best competitive ratio}
\label{sec:online_LB}

In this section, we provide several lower bounds for the best competitive ratio of an arbitrary online algorithm. 

\begin{theorem}\label{thm:neg_res_2jobs}
On general input, if there are only two jobs released, there is no online algorithm with competitive ratio better than $3/2$.
\end{theorem}

\begin{proof}
Consider an arbitrary online algorithm. Suppose that the first job is released at 0, and the algorithm replenishes and starts this job at $t$. Then, assume that the last job is released at $t+1$, therefore the algorithm schedules that job immediately, and then stops. Hence, $ALG = 2K+t+1$.

On the other hand, $OPT = \min \{2K+1,K+t+2\}$, because it either replenishes the resource once at $t+1$ or twice at $0$ and at $t+1$. 
There are two cases to consider:

\begin{enumerate}
\item If $K \leq t$, then $OPT = 2K+1$, and $ALG = 2K+t+1 \geq 3K+1$. Hence, \[ \frac{ALG}{OPT} \geq \frac{3K+1}{2K+1} \rightarrow \frac{3}{2},\text{ if } K \rightarrow \infty.\]
\item If $K > t$, then $OPT = K+t+2$, therefore,
\[\frac{ALG}{OPT} = \frac{2K+t+1}{K+t+2} \rightarrow \frac{3}{2}, \text{ if } K \rightarrow \infty.\]
\end{enumerate}

Therefore, no online algorithm can obtain a competitive ratio better than $3/2$, even if there are only two jobs released.
\end{proof}

\begin{theorem}\label{thm:neg_res_3jobs}
On general input, if there are at least three jobs released, there is no online algorithm with competitive ratio better than $4/3$.
\end{theorem}

\begin{proof}
Consider an arbitrary online algorithm. 
Suppose that the first job is released at $0$ and the algorithm replenishes and starts this job at $t_1$. 
Then, a second job is released at $t_1+1$, and the algorithm replenishes and starts this job at some $t_2\geq t_1$. 
Finally, the third and last job is released at $t_2+1$ which is replenished and started immediately. 

We can assume that the flow time of the second job is at least the flow time of the first job, i.e. $t_1+1 \leq t_2-t_1$, since replenishing and starting the second job sooner would not decrease the maximum flow time of the algorithm. Hence, $ALG = 3K+t_2-t_1$. 

On the other hand, $OPT = \min\{K+t_2+2, 2K+t_1+1, 3K+1\}$, depending on the number of replenishments (one, two or three). We are going to distinguish three cases:

\begin{enumerate}
\item If $t_1 \leq K-1$ and $t_2-t_1 \leq K-1$, then, $ALG \geq 2K+t_2+1$, and $K+t_2+2 \leq 2K+t_1+1 \leq 3K$, from which $OPT = K+t_2+2$ follows. Therefore,

\[ \frac{ALG}{OPT} \geq \frac{2K+t_2+1}{K+t_2+2} = 1 + \frac{K-1}{K+t_2+2} \geq 1+\frac{K-1}{3K} \rightarrow \frac{4}{3},\text{ if } K \rightarrow \infty.\]
\item If $t_1 \leq K-1$ and $t_2-t_1 \geq K$, then $ALG \geq 4K$, and $OPT = 2K+t_1+1$. Therefore:

\[ \frac{ALG}{OPT} \geq \frac{4K}{2K+t_1+1} \geq \frac{4K}{3K} = \frac{4}{3}.\]
\item If $t_1 \geq K$ and $t_2-t_1 \geq K$, then $ALG \geq 4K$ and $OPT = 3K+1$. Therefore,

\[ \frac{ALG}{OPT} \geq \frac{4K}{3K+1} \rightarrow \frac{4}{3},\text{ if } K \rightarrow \infty.\]
\end{enumerate}
It follows that no online algorithm can obtain a competitive ratio better than $4/3$, if there are at least three jobs released.
\end{proof}

Now we consider the $p$-regular input consisting of $n$ jobs, denoted by $R_n$. That is, $r_j = (j-1)p$ for $1 \leq j \leq n$. In Section~\ref{sec:online_general} we have presented a 2-competitive online algorithm whose competitive ratio tends to $\sqrt{2}$ as the number of jobs tends to infinity. 

In this section we investigate the question whether the above limit of $\sqrt{2}$ could be decreased to $1+\varepsilon$ for an arbitrary small $\varepsilon > 0$. 
So, we will consider only long sequences of jobs, i.e., where the number of jobs is larger than some number $n_0$, which is independent of the input.

\begin{lemma}\label{lem:q_c-comp}
On $p$-regular input, there exists $n_0 > 0$ such that for any $n \geq n_0$, the number of the replenishments in any $c$-approximate solution for $R_n$ is in 
\[
\left[\frac{1}{2c+\varepsilon_n}\sqrt{\frac{np}{K}},c\left(2\sqrt{\frac{np}{K}}+2\right)\right],
\]
where $\varepsilon_n\rightarrow 0$ as $n\rightarrow+\infty$.
\end{lemma}

\begin{proof}
From Lemma~\ref{lem:opt_lb_diff}, we have $OPT(R_n)\leq 2\sqrt{npK}+K+1$.
Hence, the objective function value in any $c$-approximate solution is at most $c(2\sqrt{npK}+K+1)$.
Since $K\geq 1$, the upper bound on the number of the replenishments immediately follows.

On the other hand, we will prove that if the number of the replenishments is too small, then the flow time of the solution is larger than the upper bound for a $c$-approximate solution. 
Suppose for contradiction that we have $q<1/(2c+\varepsilon_n)\cdot \sqrt{np/K}$, where $\varepsilon_n \rightarrow 0$ as $n$ tends to $+\infty$.
After a small transformation, we get
\begin{align*}
\frac{np}{(2c+\varepsilon_n)\sqrt{npK}}> q,
\end{align*}
and then,
\begin{align*}
\frac{np}{c(2\sqrt{npK}+K+1)+p-1}> q,
\end{align*}
if $n \geq n_0$ for some $n_0 > 0$.
We can reduce the denominator on the left-hand-side by using $OPT(R_n)\leq 2\sqrt{npK}+K+1$ again to get
\[
\frac{np}{c\cdot OPT(R_n)+p-1}> q.
\]
Rearranging terms gives
\[
np/q-p+1> c\cdot OPT(R_n).
\]
Notice that the left hand side is smaller than $qK + \lceil n/q\rceil\cdot p-p+1$, which is the cost of a schedule with $q$ replenishments by Observation~ \ref{obs:flowtime_from_replenishment}. Therefore, $q$ replenishments are not enough to obtain a $c$-approximate solution.
\end{proof}

\begin{lemma}\label{lem:F_c-comp}
On $p$-regular input, there exist a series $\varepsilon_n$ such that $\varepsilon_n\rightarrow 0$ as $n\rightarrow\infty$, and some integer $n_1 > 0$ such that for any $n \geq n_1$, the maximum flow time in any $c$-approximate solution for $R_n$ is in 
\[
\left[\frac{K^{3/2}\sqrt{(n-1)p}}{(2+\varepsilon_n)c}-p+1,c\left(2\sqrt{\frac{np}{K}}+K+1\right)\right].
\]
\end{lemma}
\begin{proof}
The upper bound on the flow time follows immediately from the upper bound of Lemma~\ref{lem:opt_lb_diff} on $OPT(R_n)$.

We proceed with the lower bound. Let $F$ be the maximum flow time of a solution with $q$ replenishments. By Proposition~\ref{prop:opt_q_repl}, we have

\[Kq + F \geq Kq + \left (\left \lceil n/q \right \rceil -1 \right )p+1.\]
After small transformations we get

\[(F+p-1)/p \geq \left \lceil n/q \right \rceil \geq n/q, \]  from which it follows that the number of the replenishments $q$ is at least $np/(F+p-1)$, thus the replenishment cost is at least $npK/(F+p-1)$.

Suppose the statement of the lemma does not hold, i.e., $F<(K^{3/2}\sqrt{(n-1)p})/((2+\varepsilon_n)c)-p+1$, for every $\varepsilon_n \rightarrow 0$ as $n \rightarrow \infty$. We will prove that then
\begin{equation}
npK/(F+p-1) > c(2\sqrt{np/K}+K+1), \label{eq:indirect}
\end{equation}
where the left hand side is a lower bound on the replenishment cost (see above), and the right hand side is an upper bound on the cost of a $c$-approximate solution (cf.~Lemma~\ref{lem:opt_lb_diff}), which is a contradiction, and the claimed lower bound on the maximum flow time follows.
To this end, we rewrite our indirect assumption:
\begin{align*}
F + p -1 < \frac{(n-1)pK}{c(2+\varepsilon_n)\sqrt{(n-1)p/K}}.
\end{align*}
Observe that for $\varepsilon_n = (K+1) /\sqrt{np/K}$, we have $\varepsilon_n \rightarrow 0$ as $n\rightarrow \infty$, and 
\begin{align*}
\frac{(n-1)pK}{c(2+\varepsilon_n)\sqrt{(n-1)p/K}}< \frac{npK}{c(2\sqrt{np/K}+K+1)},
\end{align*}
which implies (\ref{eq:indirect}).
\end{proof}

\begin{theorem}\label{thm:neg_res}
For any  $n_0> 0$, there is no deterministic online algorithm which is $1.015$ competitive on any $p$-regular input $R_n$ with $n > n_0$ even if $K=1$.
\end{theorem}
\begin{proof}
Fix any $n_0 > 0$.
Suppose  there is a $c$-competitive deterministic online  algorithm on $p$-regular input with $n \geq n_0$ jobs. 
For an arbitrary $p$-regular input $R_n$, let $(S(n), \Q(n))$ be the solution computed by the algorithm. Note that for any $n$, $R_n$ is unique, and thus $(S(n), \Q(n))$ is also uniquely defined, since the algorithm is deterministic.
 
Let $n_1>n_0$ be such that the algorithm replenishes the $2k^{th}$ time when the $n_1^{th}$ job is released at $(n_1-1)p$ for some integer $k>0$, independently whether the $n_1^{th}$ job is the last job released or not. Since the algorithm is deterministic on a $p$-regular input, $n_1$ is well-defined, and for any input where $n \geq n_1$, it produces the same schedule until $(n_1-1)p$. That is, $S(n_1)$ is a sub-schedule of $S(n)$, and $\Q(n_1) \subseteq \Q(n)$ for any $n \geq n_1$.

From Lemma~\ref{lem:F_c-comp}, we know that the maximum flow time  in  $(S(n_1),\Q(n_1))$ is at most  $U(n_1) = c\left(2\sqrt{n_1 p }+2\right)$, and the maximum flow time in $(S(n),\Q(n))$ is at least $L(n) = \sqrt{(n-1)p}/((2+\varepsilon_{n})c)-p+1$.
We can choose $n$ such that $L(n) \geq 2 U(n_1)$.
 
Now consider the following new feasible solution $(S'(n),\Q'(n))$ for $R_n$: starting with the first one, drop every second replenishment from $\Q(n)$ in $[0,(n_1-1)p]$. The flow time of the jobs arriving before $(n_1-1)p$ at most doubles (since $(n_1-1)p$ is the time of the $2k^{th}$ replenishment, it is not removed), and the flow time of the jobs released after $n_1$ does not change. Since $L(n) \geq 2 U(n_1)$, the maximum flow time of $(S'(n),\Q'(n))$ is not greater than of $(S(n),\Q(n))$.

The cost of the obtained solution is $cost(S'(n),\Q'(n)) \geq OPT(R_n)$.
However, by Lemma~\ref{lem:q_c-comp}, there are at least $\sqrt{n_1 p}/(2c+\varepsilon_{n_1})$ replenishments until $n_1$ in $Q(n)$. Therefore $cost(S'(n), \Q'(n)) \leq cost(S(n), \Q(n)) - \sqrt{n_1 p}/(4c+2\varepsilon_{n_1}) $. Thus,  $cost(S(n),\Q(n)) \geq OPT(R_n) + \sqrt{n_1 p}/(4c+2\varepsilon_{n_1})$.

Let $n_2 := 64(n_1-1)+1$. Then we have $OPT(R_{n_2})\leq 2\sqrt{n_2 p}+2$  from Lemma~\ref{lem:opt_lb_diff}, thus $OPT(R_{n_2})\leq 16\sqrt{n_1 p}$. Let $(S(n_2),\Q(n_2))$ be the schedule and replenishment structure provided by a $c$-competitive algorithm, hence,
\[cost(S(n_2),\Q(n_2)) \leq 16c\sqrt{n_1 p}.\]
On the other hand, by Lemma~\ref{lem:opt_lb},
\[
OPT(R_{n_2}) \geq 2\sqrt{n_2 p} - p + 1 = 16\sqrt{(n_1-1)p+p/64} - p + 1.
\]
 
Therefore, $(S(n_2),\Q(n_2))$ can be a $c$-approximate solution only if
\begin{align*}
\left(16+1/(4c+2\varepsilon_{n_1})\right)\sqrt{n_1 p}\leq 16c\sqrt{n_1 p}.
\end{align*}
This inequality leads to a quadratic expression in $c$, and  its solution yields that $(S(n_2),\Q(n_2))$ can be a $c$-approximate solution only if $c\geq (1+\sqrt{17/16})/2-\mu>1.015-\mu$, where $\mu\rightarrow 0$ as $n_1\rightarrow\infty$.
\end{proof}

\section{Conclusions}\label{sec:concl}

In this paper we provided a deterministic online 2-competitive algorithm for the online variant of the problem $1|jrp,s=1, p_j=1, distinct \ r_j | \Fmax$. The competitive ratio is even better for the case of $p$-regular input. Yet, there is a gap between the best upper and lower bound. The natural question arises whether it is possible to provide an online algorithm with better competitive ratio, or to derive a stronger lower bound for the best competitive ratio. There are other open questions to consider: what can we say when the jobs can have arbitrarily big processing time, or if there are multiple types of resources. These problems can be intriguing for  further research.

\section*{Acknowledgements}
This work has been supported by the National Research, Development and Innovation Office grants no. TKP2021-NKTA-01, and SNN129178. The research of P\'eter Gy\"orgyi was supported by the J\'anos Bolyai Research Scholarship of the Hungarian Academy of Sciences.


\bibliography{joint_rep}

\end{document}